\documentclass[10pt, twocolumn]{article}

\usepackage{graphicx}
\usepackage{epstopdf}
\usepackage{xcolor}
\usepackage{amsmath,amssymb,amsthm}
\usepackage[bookmarks=false]{hyperref}
\usepackage{authblk}

\newtheorem{theorem}{Theorem}

\title{Opinion Manipulation in Social Networks}
\author[1]{Alonso Silva
\thanks{Email: \href{mailto:alonso.silva@nokia-bell-labs.com}{alonso.silva@nokia-bell-labs.com}
To whom correspondence should be addressed.}}
\affil[1]{Nokia Bell Labs\\ Centre de Villarceaux\\ Route de Villejust\\ 91620 Nozay\\ France}
\date{}

\begin{document}
\maketitle

\begin{abstract}
In this work, we are interested in
finding the most efficient use of
a budget to promote an opinion
by paying agents
within a group to supplant their true opinions.
We model opinions as continuous scalars ranging from $0$ to~$1$
with $1$ ($0$) representing extremely positive (negative) opinion.
We focus on asymmetric confidence between agents.
The iterative update of an agent corresponds to the best response
to other agents' actions.
The resulting confidence matrix can be seen as
an equivalent Markov chain.
We provide simple and efficient algorithms to solve this problem
and we show through an
example how to solve the stated problem in practice.
\end{abstract}

\section{Introduction}

The process of interpersonal influence that affects agents' opinions
is an important foundation of their socialization and identity.
This process can produce shared understandings and agreements
that define the culture of the group.
The question that we are trying to answer here is
how hard or costly can it be for an external entity to change
the largest proportion of opinions of a group
by supplanting the true opinions of some agents within the group.

Starting from an initial distribution of continuous opinions
in a network of interacting agents
and agents behaving according to the best-response dynamics,
our objective is to
efficiently supplant the opinions of some agents to drive
the largest proportion of opinions towards a target set of values.
In particular, we are interested to
maximize the expected number of agents
with an opinion of at least certain threshold value.

\subsection*{Related Work}

A coordination game is played between the agents
in which adopting a common strategy has a lower cost.
When the confidence matrix is a row-stochastic matrix, it can be seen
as an equivalent Markov chain. When the Markov chain is irreducible and aperiodic,
\cite{DeGroot74} gives sufficient conditions for convergence to a consensus.
There is a growing literature on social learning using a Bayesian perspective
(see e.g.~\cite{Acemoglu2011}).
Our model belongs to the non-Bayesian framework, which keeps the computations tractable.
\cite{Yildiz2013}~studies binary $0$-$1$ opinion dynamics.
Here, we study the continuous opinion dynamics where the opinion
belongs to a bounded interval.
Our work is mostly related to~\cite{GhaderiS14} in the case of no stubborn agents.
However, in~\cite{GhaderiS14} the interactions between agents are symmetric
and the cost for each agent of differing with its interacting agents is the same.
Our work is also related to consensus problems~\cite{FagnaniZ2008}
in which the question of interest is whether
beliefs (some scalar numbers) held by different agents
will converge to a common value. 

\subsection*{Our contributions}

We study opinion dynamics
in the directed graph instead of the undirected graph.
In our opinion, this scenario is more realistic
since when an agent influences another agent it doesn't
mean that the latter influences the former.
This directed graph will be edge-weighted since
we consider different costs for an agent of differing with
each of its interacting agents.
Agents iteratively update their opinions
based on their best-response dynamics
which are given by a linear dynamic system. 
The confidence matrix describing the opinion dynamics can be seen as
an equivalent Markov chain and by decomposing
the states of this equivalent Markov chain between transient and recurrent
states,
we show that in the case we have only
recurrent states the problem
can be reduced to a knapsack problem
which can be approximated
through an FPTAS scheme.
In the presence of transient states,
the problem can be reduced
to a Mixed Integer Linear Programming problem
which in general is NP-hard but
for which there are efficient
implementations.
We show through an
example how to solve this problem
in practice.

\subsection*{Organization of the work}

The work is organized as follows.
In Section~\ref{sec:model}, we provide the
definitions and introduce the model.
In Section~\ref{sec:results},
we provide the main results of our work.
In Section~\ref{sec:example},
we give an example 
to explain how the problem can be solved in practice.
We conclude our work in Section~\ref{sec:conclusions}.
\section{Model and Definitions}\label{sec:model}

Consider a group of $n$ agents, denoted by \mbox{$\mathcal{I}=\{1,\ldots,n\}$}.
For simplicity, we consider that each agent's opinion
can be represented over the interval~$[0,1]$.
For example, they could represent
people's opinions about the current government
with an opinion $1$ corresponding to perfect satisfaction
with the current government and 
$0$ representing an extremely negative view towards the current government.
In this work, we consider a synchronous version of the problem where
time is slotted and each agent's opinion
will be given by~$x_i(t)\in[0,1]$ for $t=1,2,\ldots$
We have a budget~$B\ge 0$ and
we want to efficiently use this budget to pay some agents to favor either a positive (closest to $1$) or negative (closest to $0$) opinion
over the group of agents.
Without loss of generality (w.l.o.g.), we consider that we are promoting opinions closest to~$1$.
In the previous example, it would correspond to promote positive opinions towards the current government.
We want to supplant the opinions of
some agents in order to change the opinion of the largest proportion of the population.
We consider a threshold opinion given by~$x^*$
that we would like that the largest proportion of the population at least has.
In the previous example, it could be the threshold to have an approving or at least neutral opinion of the current government ($x^*=1/2$)
or the threshold to actually register and go to vote in the next election which we could consider to be much higher than~$1/2$ (e.g. $x^*=3/5$).
Agents who have an opinion greater or equal to the threshold are called {\sl supporters}.
If every agent is a supporter, i.e. it has an opinion greater or equal than~$x^*$, the problem is trivial
since even without spending any budget we have succeeded in achieving
our goal.
The problem gets interesting when there are agents who have opinions smaller than $x^*$.
The focus of the present work is on the asymptotic opinions of the agents.
In other words, we would like to maximize
\begin{equation*}
\lvert\{i\in\mathcal{I}: x_i(+\infty)\ge x^*\}\rvert,
\end{equation*}
where $\lvert\cdot\rvert$ denotes the set's cardinality.

We assume that there will be a cost (which will depend on the agent) for changing the agent's opinion.
In the present work, we consider that the payments have to be done at only one time~$t_0$
and without loss of generality we consider that $t_0=0$.
To differentiate between the true opinion and the expressed (after payment) opinion,
we denote the true opinion by~$\hat x_i$ and the expressed opinion by $x_i$.
We assume that the payment we need to give to agent~$i$
to change its true opinion from~$\hat x_i(0)$ to the expressed opinion~$x_i(0)$
is given by
\begin{equation}\label{eq:linearcost}
p_i=c_i(x_i(0)-\hat x_i(0)).
\end{equation}
The budget constraint is given by
\begin{equation*}
\sum_{i\in\mathcal{I}} p_i\le B.
\end{equation*}

Our objective is to solve the following problem:
\begin{equation}
\begin{aligned}
& \textrm{Maximize } \lvert\{i\in\mathcal{I}: x_i(+\infty)\ge x^*\}\rvert,\\
& \textrm{subject to }\sum_{i\in\mathcal{I}} p_i\le B,
\end{aligned}
\tag{P1}\label{opt-P}
\end{equation}
where part of the problem is to discover the dependence
between the asymptotic opinions of the agents $\{x_i(+\infty): i\in\mathcal{I}\}$
and the payments $\{p_i: i\in\mathcal{I}\}$.

We consider a weighted directed graph of the $n$ agents, denoted by $\mathcal{G}=(\mathcal{I},\mathcal{E},w)$,
where each vertex corresponds to an agent and each edge is an ordered pair of vertices $(i,j)\in\mathcal{E}\subseteq\mathcal{I}\times\mathcal{I}$ which indicates
that agent $i$ takes into account, or considers relevant, the opinion of agent~$j$.
We notice that this isn't necessarily a symmetric relationship, for this reason we consider a directed graph.

In the following, we focus on one of the agents and discuss how
this agent may change its opinion
when it is informed
of the (expressed) opinions of other agents.


We assume each agent $i\in\mathcal{I}$ has an individual cost function of the form
\begin{equation*}
J_i(x_i(t),\mathcal{N}_i)=\frac12\sum_{j\in\mathcal{N}_i}w_{ij}(x_i(t)-x_j(t))^2,
\end{equation*}
where
\mbox{$\mathcal{N}_i:=\{j\in\mathcal{I}:(i,j)\in\mathcal{E}\}$} is the neighborhood of $i\in\mathcal{I}$
and we assume that the weights $w_{ij}$ are nonnegative for all \mbox{$i,j\in\mathcal{N}$} and 
not all zeros for each $i\in\mathcal{I}$.
The objective for each agent is to minimize its individual cost function.

The above formulation defines a {\sl coordination game} with continuous payoffs~\cite{GhaderiS14}
because any vector $x=(x_1,\ldots,x_n)$ with
$x_1=x_2=\ldots=x_n$ is a Nash equilibrium.
We consider that at each time step, every agent observes the opinion of its neighbors
and updates its opinion based on these observations
in order to minimize its individual cost function.

It is easy to check that for every agent~$i\in\mathcal{I}$, its best-response strategy is given by
\begin{equation*}
x_i(t+1)=\frac1{W_i}\sum_{j\in\mathcal{N}_i}w_{ij}x_j(t),
\end{equation*}
where $W_i=\sum_{j\in\mathcal{N}_i}w_{ij}$.
We notice that this extends the work of Ghaderi and Srikant~\cite{GhaderiS14} in the case of no stubborn agents,
where they consider an undirected graph and the cost of differing to be the same across all neighbors ($w_{ij}=1$ for all $i,j$).

Define the {\sl confidence matrix} $A=[A_{ij}]$ where 
\begin{equation*}
A_{ij}=
\left\{
\begin{array}{cl}
\frac{w_{ij}}{W_i} & \textrm{if } (i,j)\in\mathcal{E},\\
0 & \textrm{otherwise.}
\end{array}
\right.
\end{equation*}
Therefore, in matrix form, the best response dynamics are given by
\begin{equation}\label{eq:BR}
x(t+1)=Ax(t),
\end{equation}
where $x(t)=(x_1(t),x_2(t),\ldots,x_n(t))$ is the {\sl vector of opinions}
at time $t$.

We notice that $A$ is a row-stochastic matrix since every element $A_{ij}$
is nonnegative and the sum of the elements in any given row is $1$.
The entry $A_{ij}$ can be interpreted as the weight (or confidence)
that agent $i\in\mathcal{I}$ gives to the opinion of agent~$j\in\mathcal{I}$.
In the following, we make the assumption that each agent has a little bit of self-confidence.

\textrm{Assumption} [Self-confidence]:
We say that the dynamical system~(\ref{eq:BR}) has {\sl self-confidence} if the diagonal of $A$ is strictly positive.
For every agent~\mbox{$i\in\mathcal{I}$}, $A_{ii}>0$ or equivalently $w_{ii}>0$.

It is assumed that the agents of the group continue to make
the revisions given by~(\ref{eq:BR}) indefinitely or until $x(t+1)=x(t)$
for some $t$ such that further revision doesn't actually
change their opinions.

Since $A$ is a row-stochastic matrix,
it can be seen as a one-step transition probability matrix
of a Markov chain with $n$ states and stationary transition probabilities.
Therefore the theory of Markov chains can be applied.

We recall some basic definitions from Markov chains which will be used afterwards.
In the following $j$ will be called a {\sl consequent}
of $i$ (of order $n$), relative to a given stochastic matrix,
if $A^n_{ij}>0$ for some $n\ge1$.
The states of the Markov chain $1,\ldots,n$ can be divided into two classes
as follows:
\begin{itemize}
\item a {\sl transient state} is one
which has a consequent of which it is not itself a consequent;
\item
a {\sl recurrent state} is one which
is a consequent of every one of its consequents.
\end{itemize}
In the following, $F$ will be the class of transient states.
The recurrent states can be further divided into
{\sl ergodic classes}, $E_1, E_2, \ldots$,
by putting two states in the same class if
they are consequent of each other (see e.g.~\cite{Doob53}, p.~179).
Then, if $i\in E_k$, $A^n_{ij}=0$ for all $j\notin E_k$, $n\ge1$.
Remember that if $i\in F$, then at least
one of its consequents lies in an ergodic class.

The decomposition of the states of the equivalent Markov chain
can be accomplished (see~\cite{DeuermeyerF1984}) in $O(\max(\lvert V\rvert,\lvert E\rvert))$.
In the following, $\mathcal{F}$ represents the class of $n_T$ transient states.
We can further decompose the class of recurrent states into
$\mathcal{E}_1,\mathcal{E}_2,\ldots,\mathcal{E}_m$
for some $m\le n$,
corresponding to the ergodic classes of the recurrent states~(see e.g.~\cite{Doob53}, p.~179).
The states of the equivalent Markov chain are aperiodic (under the self-confidence assumption).
We denote by $E_k$ the sub-matrix of $A$
representing the subgraph $\mathcal{E}_k$ of the ergodic class $k$,
composed by $n_k$ states.
Obviously, $n_1+n_2+\ldots+n_m+n_T=n$.

\section{Results}\label{sec:results}

\subsection{Dynamical Systems Without Transient States}\label{subsec:no-transient}

We focus into one of the subgraphs described by sub-matrix $E_k$.
For a subset $S\subseteq\mathcal{I}$ we denote by $1_S$
the $0/1$ vector, whose $i$-the entry is $1$ iff $i\in S$
and $(\cdot)'$ denotes the transpose operator.
Let's denote $\pi^{(k)}$ the normalized (i.e. $1_{\mathcal{E}_k}'\pi^{(k)}=1$) left eigenvector of $E_k$
associated with eigenvalue~$1$.
It is well-known (see e.g.~\cite{Doob53}, p.~182) that
the equilibrium for the ergodic class under dynamics~(\ref{eq:BR}) is unique and 
that the
agents of the ergodic class $k$ will reach a consensus (all opinions
are eventually the same) where
\begin{equation}\label{eq:recurrent}
x_i(+\infty)=\sum_{j\in\mathcal{E}_k}\pi^{(k)}_j x_j(0)\textrm{ for all }i\in\mathcal{E}_k.
\end{equation}
Therefore $\pi^{(k)}_j$ can be interpreted as the influence
of agent $j$ within its ergodic class $k$.

From eq.~(\ref{eq:linearcost}), we have that
\begin{equation*}
x_i(0)=\hat x_i(0)+\frac{p_i}{c_i}.
\end{equation*}
If we focus on ergodic class~$k$,
the problem of what is the minimum budget to
make the consensus opinion of the ergodic class to be higher than a threshold~$x^*$
becomes
\begin{equation}
\begin{aligned}
\textrm{Minimize }&P_k:=\sum_{i\in\mathcal{E}_k} p_i\\
\textrm{ subject to }&\sum_{i\in\mathcal{E}_k}\pi^{(k)}_i\left(\hat x_i(0)+\frac{p_i}{c_i}\right)\ge x^*,\\
&0\le\left(\hat x_i(0)+\frac{p_i}{c_i}\right)\le 1,\quad\forall i\in\mathcal{E}_k,\\
&\textrm{and}\quad p_i\ge 0,\quad\forall i\in\mathcal{E}_k.
\end{aligned}
\tag{P2}\label{opt-P2}
\end{equation}

Reordering the states (which can be done through any efficient sorting procedure in $O(\lvert V\rvert\log\lvert V\rvert)$ see e.g.~\cite{AhoHU1983}),
we can assume w.l.o.g. that
\begin{equation*}
\frac{\pi^{(k)}_1}{c_1}\ge\frac{\pi^{(k)}_2}{c_2}\ge\ldots\ge\frac{\pi^{(k)}_{n_k}}{c_{n_k}},
\end{equation*}
and denoting the critical item of ergodic class $k$ as
\begin{equation}
s=\min\left\{j\in\mathcal{E}_k:\sum_{i=1}^j\pi^{(k)}_i+\sum_{i=j+1}^n\pi^{(k)}_i\hat x_i(0)\ge x^*\right\},
\notag
\end{equation}
we have the following theorem:
\begin{theorem}\label{theo:max}
The optimal solution $\bar p=(\bar p_1, \bar p_2. \ldots, \bar p_{n_k})$ is given by
\begin{equation}
\begin{aligned}
\bar p_j&=
	\left\{
	\begin{array}{cl}
	c_j(1-\hat x_j(0))&\textrm{ for } j=1,\ldots,s-1,\\
	0&\textrm{ for } j=s+1,\ldots,n_k,
	\end{array}
	\right.\\
\bar p_s&=
	\frac{c_s}{\pi^{(k)}_s}\left(x^*-\sum_{j=1}^{s-1}\pi^{(k)}_j-\sum_{j=s}^{n_k}\pi^{(k)}_j\hat x_j(0)\right).\\
\end{aligned}
\notag
\end{equation}
\end{theorem}

\begin{proof}
See Appendix~\ref{app:A}.
\end{proof}

Theorem~\ref{theo:max} tell us to select the agent $i$ with the highest $(\pi_k)_i/c_i$
and to put it to the maximum possible value.
We notice that by selecting the $i$ with the highest $(\pi_k)_i/c_i$
we are selecting the agent whose influence-to-cost ratio is minimum.

From Theorem~\ref{theo:max}, the optimal value $\bar P_k=\sum_{i\in\mathcal{E}_k}\bar p_i$ of (\ref{opt-P2}) is given by
\begin{equation}\label{eq:opt-value}
\begin{aligned}
\bar P_k=&\sum_{j=1}^{s-1} c_j(1-\hat x_j(0))+\\
&+\frac{c_s}{\pi^{(k)}_s}\left(x^*-\sum_{j=1}^{s-1}\pi^{(k)}_j-\sum_{j=s}^{n_k}\pi^{(k)}_j\hat x_j(0)\right).
\end{aligned}
\end{equation}

Therefore, for each ergodic class~$k$
we have the payment $\bar P_k$, given by eq.~(\ref{eq:opt-value}),
we need to make to obtain $n_k$ agents having an opinion greater or equal than $x^*$.
More importantly, these payments~$\{\bar P_k: 1\le k\le m\}$ are 
independent between them in the sense that each payment affects only
the ergodic class where the payment was made.

In other words,
the problem~(\ref{opt-P}) is equivalent to determine~$\{z_k\}$ where
\begin{equation}
z_k=
\left\{
\begin{array}{ll}
1 & \textrm{if class $k$ is selected}\\
0 & \textrm{otherwise}
\end{array}
\right.
\end{equation}
in order to
\begin{equation}
\begin{aligned}
\textrm{Maximize}&\sum_{k=1}^m z_k n_k\\
\textrm{ subject to }&
\sum_{k=1}^m z_k \bar P_k\le B\\
\textrm{ and }&
z_k\in\{0,1\}.
\end{aligned}
\tag{P2'}\label{opt-P2-prime}
\end{equation}

This is the classic $0$-$1$ knapsack problem,
and thus we can use the well-known linear time FPTAS\footnote{
An FPTAS, short for Fully Polynomial Time Approximation Scheme, is an algorithm that for any
$\varepsilon$ approximates
the optimal solution up to an error $(1+\varepsilon)$ in time $\textrm{poly}(n/\varepsilon)$.}
algorithm of Knapsack~\cite{Lawler77} to obtain a FPTAS to problem~(\ref{opt-P}).

\subsection{Dynamical Systems With Transient States}

For the recurrent states, the previous analysis still holds.
For the transient states,
we need to use different properties of Markov chains.
From subsection~(\ref{subsec:no-transient}), we know
that the equilibrium for each ergodic class under dynamics~(\ref{eq:BR}) is unique and 
that agents within each ergodic class will reach a consensus where
\begin{equation*}
x_i(+\infty)=\sum_{j\in\mathcal{E}_k}\pi^{(k)}_j x_j(0)\textrm{ for all }i\in\mathcal{E}_k.
\end{equation*}
We denote 
\begin{equation*}
\mathcal{O}_k:=\sum_{j\in\mathcal{E}_k}\pi^{(k)}_j x_j(0)
\end{equation*}
the consensus opinion of ergodic class~$k$.
We know that the system will remain among the transient states through
only a finite number of transitions, with probability~$1$ (see e.g.~\cite{Doob53}, p.~180).
Moreover,
we have the following theorem:
\begin{theorem}\label{theo:transient}
The equilibrium for a transient state~$i\in\mathcal{F}$ (see e.g.~\cite{Doob53}, p.~182) under dynamics~(\ref{eq:BR}) is unique and given by
\begin{equation}\label{eq:transient}
x_i(\infty)=\sum_{k=1}^m h_i^{(k)}\mathcal{O}_k,
\end{equation}
where $h_i^{(k)}$ denotes the hitting probability
of the recurrent ergodic class~$\mathcal{E}_k$ starting from $i\in\mathcal{I}$.
\end{theorem}

\begin{proof}
See Appendix~\ref{app:B}.
\end{proof}

The hitting probabilities for each ergodic class can be calculated
from simple linear equations (see~\cite{Norris98}, p. 13):
\begin{equation}
\begin{aligned}
h_i^{(k)} = 1 &\quad\textrm{for } i\in \mathcal{E}_k,\\
h_i^{(k)} = \sum_{j\in\mathcal{I}} A_{ij}h_j^{(k)} &\quad\textrm {for } i\notin \mathcal{E}_k.
\end{aligned}
\label{eq:hitting}
\end{equation}

From eq.~\eqref{eq:recurrent} and eq.~\eqref{eq:transient}, we notice that the only opinions which affect the asymptotic opinions of the recurrent and the transient states
are the opinions coming from the recurrent states, therefore in the optimum the payments
will be zero for every transient state $p_t=0, \forall t\in\mathcal{F}$.
The hitting probabilities of ergodic class $k$ starting from state $i\in\mathcal{E}_j$ are
given by 
\begin{equation*}
h^{(k)}_i=\left\{
\begin{array}{ll}
1 & \textrm{if } j=k,\\
0 & \textrm{otherwise.}
\end{array}
\right.
\end{equation*}
The hitting probabilities of ergodic class $k$ starting from state $i\in\mathcal{F}$ are
calculated from eq.~(\ref{eq:hitting}).
Therefore, problem~(\ref{opt-P}) becomes
\begin{equation}
\begin{aligned}
& \textrm{Max}\sum_{i\in\mathcal{I}} I\left[\left(\sum_{k=1}^m h^{(k)}_i\sum_{i\in\mathcal{E}_k}\pi^{(k)}_i\left(\hat x_i(0)+\frac{p_i}{c_i}\right)\right)-x^*\right]\\
& \textrm{where }
I(s)=
\left\{
\begin{array}{ll}
1 & \textrm{if } s\ge0,\\
0 & \textrm{otherwise},
\end{array}
\right.\\
& \textrm{subject to the budget constraint }\sum_{i\in\mathcal{I}}p_i\le B.
\end{aligned}
\tag{P3}
\label{opt-P3}
\end{equation}

Defining
\begin{equation}
L:=\min_{i\in\mathcal{I}}\left(\sum_{k=1}^m h^{(k)}_i\sum_{i\in\mathcal{E}_k}\pi^{(k)}_i\hat x_i(0)\right),
\label{eq:lower-bound}
\end{equation}
problem~(\ref{opt-P3}) is equivalent to the following formulation:
\begin{equation}
\begin{aligned}
&\textrm{Maximize }\sum_{i\in\mathcal{I}} z_i\\
&\textrm{subject to }\\
&-\frac{\left(\sum_{k=1}^m h^{(k)}_i\sum_{i\in\mathcal{E}_k}\pi^{(k)}_i\left(\hat x_i(0)+\frac{p_i}{c_i}\right)\right)-x^*}{L-x^*}+1\ge z_i,\\
&0\le\left(\hat x_i(0)+\frac{p_i}{c_i}\right)\le 1,\quad\forall i\in\mathcal{E}_k,\\
&\sum_{i\in\mathcal{I}}p_i\le B,\\
&\textrm{and}\quad z_i\in\{0,1\}\quad\forall i\in\mathcal{I}.
\end{aligned}
\tag{P3'}
\label{opt-P3-prime}
\end{equation}

Indeed,
\begin{align*}
& I\left[\left(\sum_{k=1}^m h^{(k)}_i\mathcal{O}_k\right)-x^*\right]=0,\\
& \Leftrightarrow\quad\left(\sum_{k=1}^m h^{(k)}_i\mathcal{O}_k\right)-x^*<0,\\
&\Leftrightarrow\quad-\frac{\left(\sum_{k=1}^m h^{(k)}_i\mathcal{O}_k\right)-x^*}{L-x^*}<0,
\end{align*}
From~(\ref{opt-P3-prime}),
\begin{equation*}
-\frac{\left(\sum_{k=1}^m h^{(k)}_i\mathcal{O}_k\right)-x^*}{L-x^*}+1\ge z_i,
\end{equation*}
implies that $z_i<1$ and $z_i\in\{0,1\}$ implies that \mbox{$z_i=0$}.

Similarly,
\begin{align*}
&I\left[\left(\sum_{k=1}^m h^{(k)}_i\mathcal{O}_k\right)-x^*\right]=1,\\
&\Leftrightarrow\quad\left(\sum_{k=1}^m h^{(k)}_i\mathcal{O}_k\right)-x^*\ge0,\\
&\Leftrightarrow\quad-\frac{\left(\sum_{k=1}^m h^{(k)}_i\mathcal{O}_k\right)-x^*}{L-x^*}+1\ge1,
\end{align*}
From~(\ref{opt-P3-prime}),
\begin{equation*}
-\frac{\left(\sum_{k=1}^m h^{(k)}_i\mathcal{O}_k\right)-x^*}{L-x^*}+1\ge z_i,
\end{equation*}
implies that $z_i\le1$ but since we are maximizing the objective function we have that $z_i=1$.

The system of equations (\ref{opt-P3-prime}) is a mixed integer programming problem which is NP-hard, however
many software (see e.g. AMPL, R) can solve them quite efficiently.
In the illustrative example, we use the programming language R.

\section{Illustrative Example}\label{sec:example}

Consider as inputs:
\begin{itemize}
\item
the confidence matrix given by the coefficients in Fig.~\ref{fig:opinion},
where the self-loops were not added but their coefficients can be computed since the 
outgoing edges sum to $1$;
\item
the initial opinions of the agents 
\end{itemize}
\begin{equation}
\hat x=(0.5, 0.3,  0.4, 0.1, 0.6, 0.7, 0.3, 0.1, 0.8, 0.1, 0.2, 0.4);
\notag
\end{equation}
\begin{itemize}
\item
the cost (in dollars) to change their opinions by $+0.1$ given by 
\end{itemize}
\begin{equation}
c = (100, 80, 120, 60, 20, 100, 80, 120, 60, 20, 90, 70);
\notag
\end{equation}
\begin{itemize}
\item
the target opinion $x^*=1/2$
\item
and a budget (in dollars). 
\end{itemize}
Our objective is to determine the most efficient use of the budget to
maximize the number of agents who have an opinion of at least $x^*$ (supporters).

\subsection*{Solution}

We can decompose the system in:
\begin{itemize}
\item transient states $\mathcal{F}=\{d,e,f,g,h\}$,
\item ergodic class~$\mathcal{E}_1=\{a,b,c\}$, and
\item ergodic class~$\mathcal{E}_2=\{i,j,k,l\}$.
\end{itemize} 

The ergodic class $\mathcal{E}_1=\{a,b,c\}$ is given as in Fig.~\ref{fig:community1}.
The matrix~$E_1$ is given by the coefficients in Fig.~\ref{fig:community1}.
\begin{equation*}
E_1=
\bordermatrix{%
  &   a &   b &   c\cr
a & 0.7 & 0.3 &   0\cr
b &   0 & 0.6 & 0.4\cr
c & 0.5 &   0 & 0.5
}
\end{equation*}
The stationary distribution~$\pi^{(1)}=\left(\pi^{(1)}_a,\pi^{(1)}_b,\pi^{(1)}_c\right)'$ for ergodic class~$\mathcal{E}_1$ is
the solution of the system of equations
\begin{align*}
&(\pi^{(1)})^T E_1=(\pi^{(1)})^T,\\
&\pi^{(1)}_a+\pi^{(1)}_b+\pi^{(1)}_c=1.
\end{align*}
Therefore
\begin{equation*}
\pi^{(1)}=\left(\frac{20}{47},\frac{15}{47},\frac{12}{47}\right)'.
\end{equation*}

The ergodic class $\mathcal{E}_2=\{i,j,k,l\}$ is given as in Fig.~\ref{fig:community2}.
The matrix~$E_2$ is given by the coefficients in Fig.~\ref{fig:community2}.
\begin{equation*}
E_2=
\bordermatrix{%
  &   i &   j & k & l\cr
i & 0.8 & 0.2 & 0 & 0\cr
j & 0 &  0.8 & 0 & 0.2\cr
k & 0.2 & 0 & 0.8 & 0\cr
l & 0 & 0 & 0.2 & 0.8
}
\end{equation*}
Similarly, the stationary distribution~$\pi^{(2)}=\left(\pi^{(2)}_i,\pi^{(2)}_j,\pi^{(2)}_k,\pi^{(2)}_l\right)'$ for ergodic class~$\mathcal{E}_2$ is 
given by the solution of the system of equations
\begin{align*}
&(\pi^{(2)})^T E_2=(\pi^{(2)})^T,\\
&\pi^{(2)}_i+\pi^{(2)}_j+\pi^{(2)}_k+\pi^{(2)}_l=1.
\end{align*}
Therefore
\begin{equation*}
\pi^{(2)}=\left(\pi^{(2)}_i,\pi^{(2)}_j,\pi^{(2)}_k,\pi^{(2)}_l\right)'=\left(\frac5{39},\frac{20}{39},\frac{10}{39},\frac4{39}\right)'.
\end{equation*}

From eq.~(\ref{eq:hitting}), the hitting probabilities $h^{(1)}=\left(h^{(1)}_a,h^{(1)}_b,\ldots,h^{(1)}_l\right)$ to class $\mathcal{E}_1$ are given by
the equations 
\begin{align*}
& h^{(1)}_a=h^{(1)}_b=h^{(1)}_c=1;\\
& h^{(1)}_i=h^{(1)}_j=h^{(1)}_k=h^{(1)}_l=0;\\
\end{align*}
and
\begin{align*}
\frac35h^{(1)}_d&=\frac25+\frac15h^{(1)}_f;\\
\quad\frac25h^{(1)}_e&=\frac15+\frac15h^{(1)}_g;\\
\frac12h^{(1)}_f&=\frac3{10}h^{(1)}_d+\frac15h^{(1)}_e;\\
\quad\frac25h^{(1)}_g&=\frac15h^{(1)}_e,\\
\frac45h^{(1)}_h&=\frac15h^{(1)}_f+\frac15h^{(1)}_g.
\end{align*}
Therefore~
\begin{equation*}
h^{(1)}=\left(1,1,1,\frac{17}{18},\frac23,\frac56,\frac13,\frac7{24},0,0,0,0\right).
\end{equation*}

Since there are only two recurrent classes,
the hitting probabilities $h^{(2)}=\left(h^{(2)}_a,h^{(2)}_b,\ldots,h^{(2)}_l\right)$ to class $\mathcal{E}_2$ are given by
\begin{equation*}
h^{(2)}=\left(0,0,0,\frac{1}{18},\frac13,\frac16,\frac23,\frac{17}{24},1,1,1,1\right).
\end{equation*}

Replacing the previous quantities and solving the Mixed Integer Linear Programming problem (\ref{opt-P3-prime}) in R,
we obtain Fig.~\ref{fig:supporters} plotting the number of supporters
with respect to the budget. 
In the optimum there are only two agents who receive payments: agent~$a$ and agent~$j$.
The other agents receive zero.
The optimal payments are
\begin{center}
\begin{tabular}{|c|c|c|c|}
\hline
Budget & $p_a$ & $p_j$ & Number of supporters \\ \hline
    99 &       &    99 &  4  \\ \hline
   114 &       &   114 &  5  \\ \hline
   117 &       &   117 &  6  \\ \hline
   169 &       &   169 &  7  \\ \hline
   293 &   113 &   180 &  8  \\ \hline
   309 &   210 &    99 & 12  \\ \hline
\end{tabular}
\end{center}

\begin{center}
\begin{tabular}{|c|c|}
\hline
Number of supporters & Supporters\\ \hline
 4 & $\{i,j,k,l\}$\\ \hline
 5 & $\{h,i,j,k,l\}$\\ \hline
 6 & $\{g,h,i,j,k,l\}$\\ \hline
 7 & $\{e,g,h,i,j,k,l\}$\\ \hline
 8 & $\{e,f,g,h,i,j,k,l\}$\\ \hline
12 & $\{a,b,c,d,e,f,g,h,i,j,k,l\}$\\
\hline
\end{tabular}
\end{center}

\section{Conclusions}\label{sec:conclusions}

We have studied continuous opinion dynamics with asymmetric confidence.
The confidence matrix can be seen as
a Markov chain and by decomposing
the states between transient and recurrent
states,
we proved that in the case we have only
recurrent states the problem
can be reduced to a knapsack problem,
and in the presence of
transient states,
the problem can be reduced
to a Mixed-Integer Linear Programming problem.
We gave an illustrative example on how to solve this problem.

\bibliographystyle{hacm}
\bibliography{mybibfile}

\begin{thebibliography}{10}

\bibitem{Acemoglu2011}
{\sc Acemoglu, D., Dahleh, M.~A., Lobel, I., and Ozdaglar, A.}
\newblock Bayesian learning in social networks.
\newblock {\em The Review of Economic Studies 78}, 4 (2011), 1201--1236.

\bibitem{AhoHU1983}
{\sc Aho, A.~V., Hopcroft, J.~E., and Ullman, J.}
\newblock {\em Data Structures and Algorithms}, 1st~ed.
\newblock Addison-Wesley Longman Publishing Co., Inc., Boston, MA, USA, 1983.

\bibitem{DeGroot74}
{\sc {DeGroot}, M.~H.}
\newblock Reaching a consensus.
\newblock {\em Journal of the American Statistical Assoc. 69\/} (1974),
  118--121.

\bibitem{DeuermeyerF1984}
{\sc Deuermeyer, B.~L., and Friesen, D.~K.}
\newblock A linear-time algorithm for classifying the states of a finite markov
  chain.
\newblock {\em Operations Research Letters 2}, 6 (1984), 297 -- 301.

\bibitem{Doob53}
{\sc Doob, J.~L.}
\newblock {\em Stochastic Processes}.
\newblock John Wiley \& Sons, New York, 1953.

\bibitem{FagnaniZ2008}
{\sc Fagnani, F., and Zampieri, S.}
\newblock Randomized consensus algorithms over large scale networks.
\newblock {\em {IEEE} Journal on Selected Areas in Communications 26}, 4
  (2008), 634--649.

\bibitem{GhaderiS14}
{\sc Ghaderi, J., and Srikant, R.}
\newblock Opinion dynamics in social networks with stubborn agents: Equilibrium
  and convergence rate.
\newblock {\em Automatica 50}, 12 (2014), 3209--3215.

\bibitem{Lawler77}
{\sc Lawler, E.~L.}
\newblock Fast approximation algorithms for knapsack problems.
\newblock In {\em Proceedings of the 18th Annual Symposium on Foundations of
  Computer Science\/} (Washington, DC, USA, 1977), SFCS '77, IEEE Computer
  Society, pp.~206--213.

\bibitem{Norris98}
{\sc Norris, J.~R.}
\newblock {\em Markov chains}.
\newblock Cambridge series in statistical and probabilistic mathematics.
  Cambridge University Press, 1998.

\bibitem{Yildiz2013}
{\sc Yildiz, M.~E., Ozdaglar, A.~E., Acemoglu, D., Saberi, A., and Scaglione,
  A.}
\newblock Binary opinion dynamics with stubborn agents.
\newblock {\em {ACM} Trans. Economics and Comput. 1}, 4 (2013), 19.

\end{thebibliography}

\appendix
\section{Appendix A}\label{app:A}
\subsection*{Proof of Theorem~\ref{theo:max}}
\begin{proof}
Any optimal solution~$p=(p_1,p_2,\ldots,p_{n_k})$ must be maximal in the sense that 
\begin{equation*}
\sum_{i\in\mathcal{E}_k}\pi^{(k)}_i\left(\hat x_i(0)+\frac{p_i}{c_i}\right)=x^*.
\end{equation*}
Assume without loss of generality that 
\begin{equation*}
\frac{\pi^{(k)}_j}{c_j}>\frac{\pi^{(k)}_{j+1}}{c_{j+1}}\quad\forall j\in\mathcal{E}_k
\end{equation*}
and let $p^*$ be the optimal solution of~(\ref{opt-P2}).
Suppose, by absurdity, that for some $\ell<s$, \mbox{$p^*_\ell<c_\ell(1+\hat x_\ell(0))$},
then we must have $p_q^*>\bar p_q$ for at least one item $q\ge s$.
Given a sufficiently small $\varepsilon>0$, we could increase the value of $p^*_\ell$
by $\varepsilon$ and decrease the value of $p^*_q$ by $\varepsilon\frac{\pi^{(k)}_\ell}{c_\ell}\frac{c_q}{\pi^{(k)}_q}$,
thus diminishing the value of the objective function by $\varepsilon\left(\frac{\pi^{(k)}_\ell}{c_\ell}\frac{c_q}{\pi^{(k)}_q}-1\right)>0$
which is a contradiction.
Similarly, we can prove that $p^*_\ell>0$ for $\ell>s$ is impossible. 
Hence $\bar p_s=\frac{c_s}{\pi^{(k)}_s}\left(x^*-\sum_{j=1}^{s-1}\pi^{(k)}_j-\sum_{j=s}^{n_k}\pi^{(k)}_j\hat x_j(0)\right)$
for maximality.
\end{proof}

\section{Appendix B}\label{app:B}
\subsection*{Proof of Theorem~\ref{theo:transient}}
\begin{proof}
We first recall the definition of hitting probabilities.
Let $(X_n)_{n\ge0}$ be a Markov chain with transition matrix $A$.
The first hitting time of a subset $\mathcal{E}\subseteq\mathcal{I}$
is the random variable 
\begin{equation*}
\tau^\mathcal{E}(w)=\inf\{n\ge0: X_n(w)\in\mathcal{E}\},
\end{equation*}
where we agree that the infimum of the empty set is $+\infty$.
The hitting probability starting from $i$ that $(X_n)_{n\ge0}$ ever hits $\mathcal{E}$ is given by
\begin{equation}\label{eq:defhit}
h_i^\mathcal{E}=\mathbb{P}_i(\tau^\mathcal{E}<+\infty).
\end{equation}
In order to simplify the notation,
we denote the hitting probability of ergodic class $k$ as
\begin{equation*}
h_i^{(k)}:=h_i^{\mathcal{E}_k}.
\end{equation*}
Under the self-confidence assumption, if $j\in\mathcal{E}_k$ we have that (see e.g.~\cite{Doob53}, p.~180)
\begin{equation}\label{eq:asym}
\lim_{\ell\to+\infty} A^\ell_{ij}=\rho_i(\mathcal{E}_k)\pi^{(k)}_j
\end{equation}
where
\begin{equation*}
\rho_i(\mathcal{E}_k)=\lim_{n\to+\infty}\sum_{j\in\mathcal{E}_k} A^n_{ij}.
\end{equation*}
Considering the equivalent Markov chain, we have that
\begin{align}
\rho_i(\mathcal{E}_k)&=\lim_{n\to+\infty}\sum_{j\in\mathcal{E}_k}\mathbb{P}_i(X_n(w)\in j)\notag\\
&=\lim_{n\to+\infty}\mathbb{P}_i(X_n(w)\in\mathcal{E}_k).\label{eq:macondo}
\end{align}
We have that the following equality of sets 
\begin{equation}\label{eq:setsequality}
\{w:X_n(w)\in\mathcal{E}_k\}=\{w:\tau^{\mathcal{E}_k}(w)\le n\}.
\end{equation}
Therefore
\begin{align*}
\rho_i(\mathcal{E}_k)&=\lim_{n\to+\infty}\mathbb{P}_i(X_n(w)\in\mathcal{E}_k)\\
&=\lim_{n\to+\infty}\mathbb{P}_i(\tau^{\mathcal{E}_k}(w)\le n)\\
&=\mathbb{P}_i(\tau^{\mathcal{E}_k}(w)<\infty)\\
&=h_i^{(k)}
\end{align*}
where the first equality is coming from eq.~\eqref{eq:macondo},
the second equality from eq.~\eqref{eq:setsequality},
and the last equation from eq.~\eqref{eq:defhit}.
Replacing in eq.~\eqref{eq:asym}, if $j\in\mathcal{E}_k$
\begin{equation*}
\lim_{\ell\to+\infty} A^\ell_{ij}=h_i^{(k)}\pi^{(k)}_j.
\end{equation*}

Therefore
\begin{align*}
&x_i(+\infty)=\lim_{\ell\to+\infty}x_i(\ell+1)\\
&=\lim_{\ell\to+\infty}\sum_{j\in\mathcal{I}}A_{ij}^\ell x_j(0)\\
&=\lim_{\ell\to+\infty}\sum_{j\in\mathcal{F}}A_{ij}^\ell x_j(0)+
\lim_{\ell\to+\infty}\sum_{k=1}^m\sum_{j\in\mathcal{E}_k}A_{ij}^\ell x_j(0)\\
&=\sum_{k=1}^m\sum_{j\in\mathcal{E}_k}\lim_{\ell\to+\infty}A_{ij}^\ell x_j(0)\\
&=\sum_{k=1}^mh_i^{(k)}\sum_{j\in\mathcal{E}_k}\pi_j^{(k)}x_j(0)\\
&=\sum_{k=1}^mh_i^{(k)}\mathcal{O}_k.
\end{align*}
\end{proof}

\onecolumn
\section*{Figures}

\begin{figure}[h!]
\caption{Opinion Dynamics}\label{fig:opinion}
\centering
	\includegraphics[width=.9\textwidth]{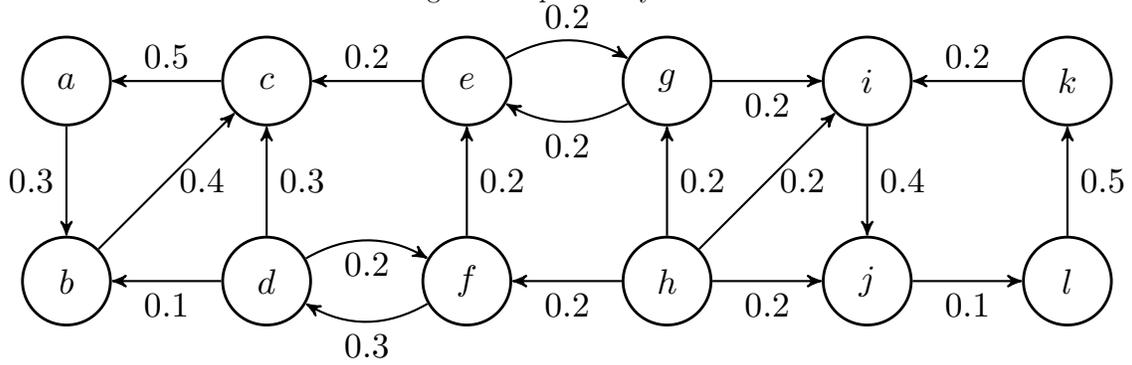}
\end{figure}

\begin{figure}[h!]
\caption{Ergodic class $\mathcal{E}_1$}\label{fig:community1}
\centering
	\includegraphics[width=.3\textwidth]{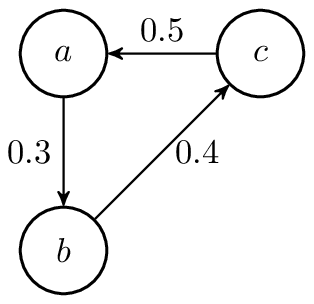}
\end{figure}

\begin{figure}[h!]
\caption{Ergodic class $\mathcal{E}_2$}\label{fig:community2}
\centering
	\includegraphics[width=.3\textwidth]{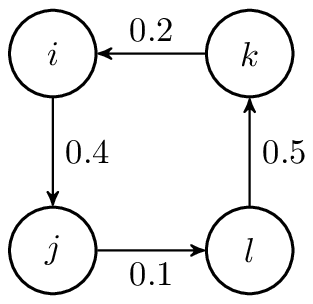}
\end{figure}

\begin{figure}[h!]
\caption{Number of supporters vs budget}\label{fig:supporters}
\centering
	\includegraphics[width=.9\textwidth]{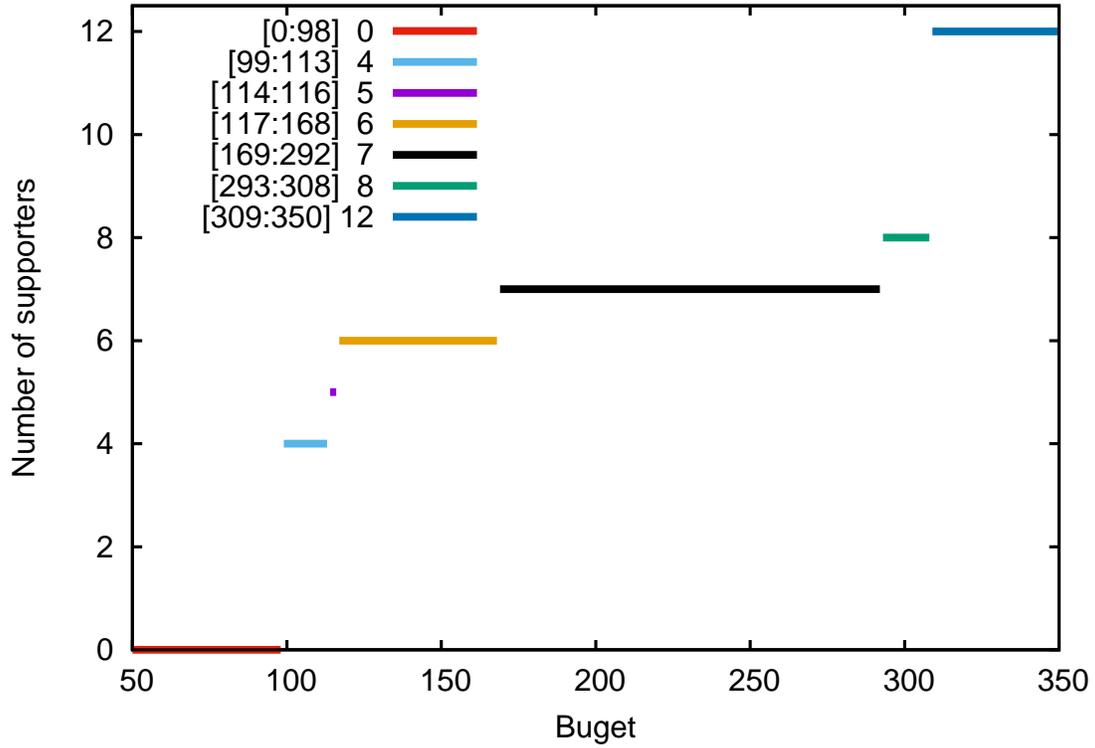}
\end{figure}

\newpage 
Inputs:
\begin{itemize}
\item The confidence matrix
	\begin{equation*}
	A=\bordermatrix{%
	  &  a &   b &	 c &   d &   e &   f &   g &   h &   i &   j &   k &   l \cr
	a &0.7 & 0.3 &   0 &   0 &   0 &   0 &   0 &   0 &   0 &   0 &   0 &   0 \cr
	b &  0 & 0.6 & 0.4 &   0 &   0 &   0 &   0 &   0 &   0 &   0 &   0 &   0 \cr
	c &0.5 &   0 & 0.5 &   0 &   0 &   0 &   0 &   0 &   0 &   0 &   0 &   0 \cr
	d &  0 & 0.1 & 0.3 & 0.4 &   0 & 0.2 &   0 &   0 &   0 &   0 &   0 &   0 \cr
	e &  0 &   0 & 0.2 &   0 & 0.6 &   0 & 0.2 &   0 &   0 &   0 &   0 &   0 \cr
	f &  0 &   0 &   0 & 0.3 & 0.2 & 0.5 &   0 &   0 &   0 &   0 &   0 &   0 \cr
	g &  0 &   0 &   0 &   0 & 0.2 &   0 & 0.6 &   0 & 0.2 &   0 &   0 &   0 \cr
	h &  0 &   0 &   0 &   0 &   0 & 0.2 & 0.2 & 0.2 & 0.2 & 0.2 &   0 &   0 \cr
	i &  0 &   0 &   0 &   0 &   0 &   0 &   0 &   0 & 0.6 & 0.4 &   0 &   0 \cr
	j &  0 &   0 &   0 &   0 &   0 &   0 &   0 &   0 &   0 & 0.9 &   0 & 0.1 \cr
	k &  0 &   0 &   0 &   0 &   0 &   0 &   0 &   0 & 0.2 &   0 & 0.8 &   0 \cr
	l &  0 &   0 &   0 &   0 &   0 &   0 &   0 &   0 &   0 &   0 & 0.5 & 0.5  
	};
	\end{equation*}
\item The initial true opinions of the agents
	\begin{equation*}
	\hat x=
	\bordermatrix{%
	&   a &  b &	 c &  d &  e &   f &  g &   h &  i &  j &  k & l \cr
	& 0.5 & 0.3 &  0.4 & 0.1 & 0.6 & 0.7 & 0.3 & 0.1 & 0.1 & 0.8 & 0.2 & 0.4 
	};
	\end{equation*}
\item The costs to change their initial true opinions by +0.1
	\begin{equation*}
	c=
	\bordermatrix{%
	&   a &  b &	 c &  d &  e &   f &  g &   h &  i &  j &  k & l \cr
	& 100 & 80 & 120 & 60 & 20 & 100 & 80 & 120 & 60 & 20 & 90 & 70 
	};
	\end{equation*}
\item The target opinion $x^*=1/2$;
\item A budget $B$ (in dollars).
\end{itemize}

\end{document}